\documentclass[jmp,preprint,12pt]{revtex4-1}
\usepackage{amsmath,amssymb,amsthm,geometry,bbm,calrsfs,graphicx}
\usepackage{dcolumn}
\usepackage{bm}
\usepackage[utf8]{inputenc}
\usepackage[T1]{fontenc}
\usepackage{rsfso}
\usepackage{natbib}
\usepackage{color}
\usepackage{cancel}
\usepackage{diagbox}
\usepackage{hyperref}
\usepackage[normalem]{ulem}
\usepackage{tikz-cd}
\newcommand{\ket}[1]{\big|#1\rangle}
\newcommand{\bra}[1]{\langle #1\big|}
\newcommand{\state}[1]{\ket{#1}\bra{#1}}
\newcommand{\expc}[1]{\left\langle#1\right\rangle}
\newcommand{\hilb}{\mathcal{H}}

\newcommand{\Tr}{\mathrm{Tr}}
\newcommand{\ident}{\mathbbm{1}}
\newcommand{\rr}{\mathbb{R}}
\newcommand{\cx}{\mathbb{C}}
\newcommand{\Expc}{\mathbb{E}}

\newcommand{\dint}[2]{{\displaystyle \int_{#1}^{#2}}}
\newcommand{\norm}[1]{\left\lVert #1\right\rVert}
\newcommand{\lind}{\mathcal{L}}
\newcommand{\filt}{\mathcal{F}}
\newcommand{\EE}{\mathcal{E}}
\newcommand{\fock}{\Gamma(\mathfrak{h})}

\newcommand{\wt}[1]{\widetilde{#1}}
\newcommand{\normpsi}{\left\langle\psi_t|\psi_t\right\rangle}
\newcommand{\inormpsi}{\left\langle\psi_t|\psi_t\right\rangle^{-1}}
\newcommand{\PN}{\mathbb{P}_{\wt{N}}}
\newtheorem{proposition}{Proposition}
\newtheorem{theorem}{Theorem}

\newtheorem{lemma}{Lemma}
\newtheorem{corollary}{Corollary}
\begin{document}
\author{Dustin Keys}
\affiliation{Program in Applied Mathematics, University of Arizona, Tucson, Arizona, 85721, USA}
\author{Jan Wehr}
\affiliation{Department of Mathematics and Program in Applied Mathematics, University of Arizona, Tucson, Arizona, 85721, USA}
\date{\today}
\title{Poisson stochastic master equation unravellings and the measurement problem:\\
a quantum stochastic calculus perspective}

\begin{abstract}
The paper studies a class of quantum stochastic differential equations, modeling an interaction of a system with its environment in the quantum noise approximation.  
The space representing quantum noise is the symmetric Fock space over $L^2\left(\mathbb{R}_+\right)$.  Using the isomorphism of this space with the space of square-integrable functionals 
of the Poisson process, the equations can be represented as classical stochastic differential equations, driven by Poisson processes.  This leads to a discontinuous dynamical state 
reduction which we compare to the Ghirardi-Rimini-Weber model. A purely quantum object, the norm process, is found which plays the role of an observer (in the sense of Everett
[H. Everett III, Reviews of modern physics, 29.3, 454, (1957)]), encoding all events occurring in the system space. An algorithm introduced by Dalibard et al 
[J. Dalibard, Y. Castin, and K. Mølmer, Physical review letters, 68.5, 580 (1992)] to numerically solve quantum master equations is 
interpreted in the context of unravellings and the trajectories of expected values of system observables are calculated.
\end{abstract}
\maketitle
\section{Introduction}
In the strict sense, every quantum system is an open quantum system, with the only exception being the entire universe. A system can be isolated from its environment to
a certain degree but only temporarily. In order to model such an open system, the environment must be approximated in some manner, as it is simply impossible to model exactly.
If the coupling between the system and the environment is such that the correlation time scale of the environment is much smaller than the characteristic time scale of the
system, then the system can be modeled as Markovian and the reduced system dynamics can be described by a quantum dynamical semigroup acting on the density operator of the system. This semigroup
is generated by an operator called the Lindbladian which, in turn, defines a differential equation called the GKSL master equation, after Gorini, Kossakowski, Sudarshan and
Lindblad \cite{GKS76,Lind76}.  This equation is of great utility in the study of open quantum systems, as it governs the time evolution of the reduced density operator and thus
of expected values of all system observables. In the 1990's, Monte Carlo methods were developed, facilitating numerical integration of the master
equation \cite{DCM92,GP92}.  Their mathematical foundations are stochastic processes with values in the system's Hilbert space, whose expected value at every time is the density operator of the 
GKSL equation at that time. 
This stochastic process is called an unravelling of the master equation. It is also known as the quantum trajectories model.  The computational method for integrating the GKSL equation 
based on its unravelling is useful, since, instead of following the evolution of a density operator, 
one averages several realizations of the evolution of  a state \emph{vector}, thus significantly reducing the number of dimensions involved. But a stochastic unravelling itself also has an 
interpretation beyond
that of an algorithm for integrating the master equation, namely that of a random trajectory of the system, where the influence of the environment is modeled via the introduction of noise.
This can be a conceptual tool in the study of open quantum systems and, in addition, offers a way to think of the measurement problem.  
Depending on the interpretation chosen, the noise may be thought of as coming from the environment, or as an irreducible factor in the evolution of \emph{every} quantum system.
\par
The measurement problem arises from the tension between the dynamical equation describing the evolution of a quantum system, the Schr\"odinger equation (SE), which tends to produce
superpositions of possible outcomes, and the observed fact that macroscopic superpositions are never actually observed. In practice, the result of a measurement is always a single well
defined outcome---one of several possible results of an experiment performed on the system, whose initial state may be their superposition. There
is a rule---the Born rule---which describes the probability of any outcome given the system's initial state. An additional postulate---the projection postulate (PP)---says that the state of the 
system is projected onto the state corresponding to the outcome of the measurement. This amounts to saying that any quantum
system obeys two dynamical rules---SE describing a continuous, deterministic evolution, and PP describing an indeterministic projection onto the outcome of a measurement.
To describe most laboratory systems, these two dynamical rules, viewed as an algorithm for describing the outcome of a measurement, are sufficient. However, when viewed as a foundational basis for 
all subsequent physical theories, they leave much to be desired.  At least a part of the problem lies in a rigorous description of the process
of measurement, on which there have been many ideas but little agreement. To some, a measurement is something performed by an observer, thought of either as a conscious entity or
at least something complex enough to act on the system in a certain manner. In either case, there is no hard line between when a system may be thought of as an observer and when
it is simply another quantum system, which means that at the intermediary stages of complexity (or consciousness) this idea loses consistency and so is insufficient as a 
fundamental description of reality. To get around the issue of defining the concepts of measurement or observers, stochastic localization theories were proposed \cite{GRW85,GPR90}, known as
GRW (after Ghirardi, Rimini and Weber) which describes discontinuous localization, or CSL (continuous spontaneous localization) in which localization occurs continuously in time. This
solves the issue by positing that the collapse happens at the fundamental level, involving neither observer nor measurement, and indeterministically, in a manner which is closely
related to the quantum trajectories picture of open quantum systems. Their dynamical equation, in both cases, is a stochastic nonlinear differential equation.  An effort
was put into showing that these features were necessary, purporting to demonstrate that there was no hope for SE to describe collapse  \cite{BG00}.
\par
In the present work, we show that GRW-like dynamics can be achieved by using a unitary dynamical equation together with a normalization procedure, which does not 
introduce any independent dynamical elements, but which conditions the resulting trajectories on the initial condition of the system.
In order to do this, we model interaction with the environment, using \emph{quantum noise}.  Quantum noises are a part of  quantum stochastic calculus (QSC), a theory
which was put into fully mature mathematical form by Hudson and Parthasarathy in the mid 1980's \cite{HudPar84}.  This theory models the environment of a quantum system as a bosonic 
(i.e. symmetric) Fock space over the Hilbert space $L^2\left(\mathbb{R}_+\right)$ and uses four fundamental families of operators (processes) in this space: creation, annihilation, conservation
(or scattering) and time, to drive quantum stochastic differential equations.  Hudson and Parthasarathy identified a class of quantum stochastic differential equations---known as HP equations---which 
describe unitary quantum stochastic processes.  Given a Lindbladian, coefficients of an HP equation can be chosen so that the reduced dynamics of the system, obtained by taking the trace over the 
quantum noise space, obeys the Lindblad equation. A detailed account of the relation of these HP equations to modes of measurement which can be used in laboratory situations was laid out by Barchielli
\cite{Bar06}. 
\par
As is well known, the bosonic Fock space over $L^2\left(\mathbb{R}_+\right)$ is isomorphic to the Hilbert space of square-integrable functions on the Wiener space---via the Wiener chaos decomposition
(see e.g. \cite{Nual06}). 
Several other isomorphisms of this space to $L^2$ spaces on classical probability spaces exist, including the space of square-integrable functionals of the Poisson process.  Using these isomorphisms,  
quantum SDE can be mapped to classical ones, in particular leading to equations driven by Wiener processes (as in CSL), or
by Poisson processes (as in GRW).  One first defines a unitary dynamics on the tensor product of the system's space and the Fock space, 
and then uses one of the isomorphisms mentioned above to obtain a representation of the dynamics as a random classical one.  We thus obtain this way a stochastic unravelling of a master equation
starting from the quantum noise model.  This can be interpreted as deriving stochastic dynamics from a fully quantum model, including models of GRW or CSL type from the quantum noise theory.  
With an appropriate choice of the coefficients of the quantum SDE (equivalent to a choice of the reduced Lindblad dynamics), such a  model leads to spontaneous localization (state reduction) 
during a measurement process, in full quantitative agreement with the Born rule and the projection postulate, but the model presents an object which is part of the total quantum state
and which encodes the history of events taking place in the system, acting as a kind of system observer.  We emphasize that the resulting classical SDE 
for the system dynamics are \emph{nonlinear} when written in terms of a properly normalized system state vector, but when the total state is considered, including the noise,
linearity is restored.
\par
The procedure outlined above has been realized, using the isomorphism of the Fock space with the $L^2$ space over the Wiener space by Parthasarathy and Usha Devi \cite{ParDevi17}, although
previous work was done by Belavkin, Barchielli, and Staszewski \cite{Bel90,BelStas91,BarBel91}.  In this paper we 
focus on the Poisson representation.  In addition to providing a derivation of a model of spontaneous state reduction, our results provide a quantum stochastic context for 
the Monte Carlo algorithm, used to numerically integrate Lindblad master equations. 
The paper is organized as follows.  Section 2 deals with the GKSL equation and describes the unravellings; in Section 3 we discuss Poisson stochastic calculus
which is used in GRW.  Quantum stochastic calculus is introduced in section 4, with section 5 describing how to obtain classical Poisson stochastic calculus via an isomorphism between
square integrable functionals of the Poisson process and the Fock space. Section 6 contains the derivation of unravellings from the HP equation, with section 7 expanding on the Girsanov
transformation used to maintain normalization of the unravelling process. Section 8 deals with applications related to the stochastic Schr\"odinger equation, including the procedure for
obtaining GRW, and interprets the Monte Carlo 
algorithm mentioned previously in the context of stochastic calculus and contains a calculation of the dynamics of expected values of observables, and in section 9 we comment more specifically on the 
relation of our results to the GRW model.  Concluding remarks are the content of section 10.

\section{The GKSL equation and its unravellings}
The GKSL equation can be derived from physical considerations (see e.g. \cite{BP02}), as the equation which governs the reduced evolution of a quantum systems in
contact with the environment, if a Born-Markov assumption is made about the coupling with the environment. The joint dynamics takes place in the total Hilbert space,
$\wt{\hilb}=\hilb_S\otimes\hilb_E$, where $\hilb_S$ is the system Hilbert space, e.g. $\cx^n$, and $\hilb_E$ is an infinite dimensional Hilbert space to be further
specified later. The joint state $\rho(t)$ has a dynamics generated by the total Hamiltonian $H=H_0+H_I$, with $H_0$ the uncoupled Hamiltonian and $H_I$
the interaction Hamiltonian. The Born-Markov assumption is equivalent to requiring that the relaxation time scales of the subsystem must be long compared with the correlation
time of the environment.  The reduced system state is defined in terms of a density operator, $\rho_S$, obtained by a partial trace over the degrees
of freedom associated with the environment and its dynamics comes from the unitary evolution, $U(t)=\exp(-itH)$ (setting $\hbar=1$), of the total density matrix, 
$\rho(t)=U(t)\rho_0U^\dagger(t)$, with $\rho_0$ the initial density matrix, via this partial trace, $$\rho_S(t)=\Tr_E\left[U(t)\rho_0 U^\dagger(t)\right].$$
This dynamics can be expressed, in the Born-Markov approximation, as the action of a quantum dynamical semigroup, $\rho_S(t)=T_t[\rho_S(0)]$, which is generated by an operator, 
$\lind$, sometimes called a Lindbladian. The explicit form for the generator was described by Gorini, Kossakowski, Sudarshan and Lindblad \cite{GKS76,Lind76} who showed that
\begin{equation}
    \lind[\rho]=-i[H_{LS},\rho]+\sum_\alpha L_\alpha\rho L_\alpha^\dagger-\frac{1}{2}\left\{ L_\alpha^\dagger L_\alpha,\rho \right\}.
\end{equation}
where $H_{LS}$ is a self-adjoint operator which is a relaxed Hamiltonian accounting for the system Hamiltonian but with some change in energy levels due to the interaction with the
environment which is called the Lamb shift. The $L_\alpha$'s are bounded operators called Lindblad operators, and the sum 
$\sum_\alpha L_\alpha^\dagger L_\alpha$ must be strongly convergent to a bounded operator. These Lindblad operators can be derived from the interaction Hamiltonian system operators,
if these operators are decomposed into eigenoperators of the system Hamiltonian, and the bath operators coupled to them are modeled through their correlations as coupling coefficients
comprising a matrix of coefficients which is then diagonalized \cite{BP02}. Then the reduced density operator follows a differential equation, derivable from Schr\"odinger dynamics with
appropriate approximations, which is called the
GKSL master equation
\begin{equation}
    \dfrac{d}{dt}\rho_S(t)=\lind[\rho_S(t)].
\end{equation}
\par
Monte Carlo methods exist for integrating this equation \cite{DCM92,GP92}, which in the case of large dimensions can vastly improve on the approach of integrating the 
matrix differential equation. These methods, called unravellings, employ stochastic processes, either Poisson, $N(t)$, or Wiener, $W(t)$, with respect to which a stochastic integral can be 
defined.  This allows to define stochastic differentials and express relations between stochastic processes in the form of stochastic differential equations (SDE), which are manipulated using 
stochastic calculus rules---the  It\^o rules. 
An unravelling with respect to a stochastic process $M(t)= N(t),W(t)$, is a Hilbert space-valued stochastic process $\ket{\psi_t(M)}$, such that if we let 
$\rho_t=\state{\psi_t(M)}$ be the pure state density-matrix-valued stochastic process, then we have
$$\Expc d\rho_t=\lind[\Expc\rho_t]dt.$$
The process $\ket{\psi_t}$ satisfies a stochastic differential equation, $d\ket{\psi_t}=G(\ket{\psi_t})dt+F(\ket{\psi_t})dM(t)$ which is an abbreviated differential notation for the stochastic 
integral equation
$\ket{\psi_t}=\dint{0}{t}G(\ket{\psi_s})ds+\dint{0}{t}F(\ket{\psi_s})dM(s)$.
The Gisin-Percival equation \cite{GP92} is one such unravelling, satisfying the SDE
\begin{align}
    d\ket{\psi_t}=-iH\ket{\psi_t}dt+&\sum_\alpha\left( \expc{L_\alpha^\dagger}_{\psi_t}L_\alpha-\frac{1}{2}L^\dagger_\alpha L_\alpha-\frac{1}{2}|\expc{L_\alpha}_{\psi_t}|^2
    \right)\ket{\psi_t}dt\nonumber\\
    &+\frac{1}{\sqrt{2}}\sum_\alpha \left(L_\alpha-\expc{L_\alpha}_{\psi_t}\right)\ket{\psi_t}dW_\alpha(t)
\end{align}
where $dW_\alpha$ are complex valued Wiener processes satisfying the It\^o rules (with $dW_\alpha^*$ denoting the complex conjugate of $dW_\alpha$) ,
\begin{align*}
    dW_\alpha(t) dW_\beta(t)=dW_\alpha^*(t)dW_\beta^*(t)=0,\\
    dW_\alpha^*(t)dW_\beta(t)=2\delta_{\alpha\beta}dt,
\end{align*}
and $\expc{L_\alpha}_{\psi_t}=\expc{\psi_t|L_\alpha|\psi_t}$, making the Gisin-Percival equation non-linear. Another such unravelling is driven by Poisson processes $N_\alpha(t)$, 
\begin{align}
    d\ket{\psi_t}=-\bigg(iH&+\frac{1}{2}\sum_\alpha L_\alpha^\dagger L_\alpha-\expc{L_\alpha^\dagger L_\alpha}_{\psi_t}\bigg)\ket{\psi_t}dt\nonumber\\
    &+\sum_\alpha \left(\dfrac{L_\alpha}{\langle L_\alpha^\dagger L_\alpha\rangle_{\psi_t}^{1/2}}-I\right)\ket{\psi_t}dN_\alpha(t).
    \label{PDP}
\end{align}
This process, $\ket{\psi_t}$, is called a piecewise deterministic process (PDP) by Breuer and Petruccione \cite{BP02}. They were written in this form by Barchielli and Belavkin
\cite{BarBel91}, though they were derived earlier from quantum stochastic considerations by Belavkin \cite{Bel90}, in the context of non-demolition measurements. The It\^o
rule for Poisson processes is 
\begin{equation*}
    dN_\alpha(t) dN_\beta(t)=\delta_{\alpha\beta}dN_\alpha(t).
\end{equation*}
In both cases, $\ket{\psi_t}$ can be shown to be an unravelling of the GKSL master equation by using the It\^o rules, together with the product rule
$$d(\state{\psi_t})=(d\ket{\psi_t})\bra{\psi_t}+\ket{\psi_t}d\bra{\psi_t}+d\ket{\psi_t}d\bra{\psi_t},$$
and the expectations $\Expc dW_\alpha(t)=0$, $\Expc dN_\alpha(t)=\Expc \expc{L^\dagger_\alpha L_\alpha}_{\psi_t}dt$. Typically, the standard Poisson process has expected value
$\Expc dN_\alpha(t)=dt$, here a change of distribution was used so that the PDP describes an unravelling of the GKSL equation. We will elaborate more on this later.

\section{Poisson measures and their stochastic integrals}

The well studied Wiener process has a property called the chaotic decomposition property, which was first explored by Wiener \cite{Wiener38} and developed further by It\^o \cite{Ito56}. 
Now it constitutes the core of the Malliavin calculus (see e.g. \cite{Nual06}), which as we will see is closely related to quantum stochastic calculus. The basic object is a white noise measure,
$W(\omega,A)$, which to every set $A\in \mathcal{T}$, with $\mathcal{T}$ a $\sigma$-algebra, in some state space $(T,\mathcal{T})$ assigns a mean zero Gaussian random variable, $W(A)$, 
defined on a probability space 
$(\Omega,\filt,\mathbb{P})$, where $\omega\in\Omega$, 
with variance given by an intensity measure $\mu$ such that $\Expc W(A)W(B)=\mu(A\cap B)$. Although the Malliavin calculus for Poisson processes, which will concern us
in what follows, had long been less developed there has been work done to remedy that and Pecatti and Reitzner have compiled some of this in a recent book \cite{PecReit16}. We will
briefly introduce the Malliavin calculus as they did, in particular the chapter by Last, and refer the reader to their book for more details.
\par
Analogously to the Wiener measure, we may think of a point process as a random measure $\xi(\omega,A)$ mapping sets $A$ in state space $(T,\mathcal{T})$, to a random variable $\xi(A)$
taking values in $\mathbb{Z}\cup\{\infty\}$ and defined on probability space $(\Omega,\filt,\mathbb{P})$, with $\omega\in\Omega$.  For this measure, we define a
real valued intensity measure $\nu$ such that
$$\nu(A)=\Expc \xi(A),$$
and if we impose the probability law
$$\mathbb{P}(\xi(A)=k)=\dfrac{\nu(A)^k}{k!}\exp(-\nu(A)),$$
and require that $\{\xi(A_i)\}$ are independent random variables for pairwise disjoint sets $A_1,\ldots,A_n$, then the random measure, $N=\xi$, is called Poissonian. If
$\nu$ is the Lebesgue measure on $T=\rr_+$ then it is the standard Poisson process. A related but important process is the compensated Poisson process, $\wt{N}=N-\nu$, which is a mean
zero martingale.
Point processes can sometimes be described in terms of random variables, $X_i$, with values in $\rr_+$ and distributed according to some distribution $\mathbb{Q}$. 
For example, a simple jump process is defined for a random element $X\in\rr_+$ as the delta measure of that random element, $\delta_X(B)=\ident_B(X)$ indicating a jump of one 
when $X\in B$. We may also add these simple jump processes to create new processes. So for example the Poisson process can be realized as a sum $N=\sum_{i=1}^{n(\omega)}\delta_{X_i}$, where
$n(\omega)$ is random variable which is Poisson distributed with parameter $\lambda$. In this case we have $\nu=\lambda\mathbb{Q}$. A product measure may be defined on $T^m$ as
$$\xi^{(m)}(C)=\sum_{i_1\neq\ldots\neq i_m<\xi(\mathbb{X})}\null\ident_{C}(X_{i_1},\ldots,X_{i_m}),$$
where pairwise different $i_j$ are taken. 
\par
Let $\mathbf{F}$ be the set of functionals taking random measures to real numbers. Then we define the difference operator, $D_x$ with $x\in T$, acting on $f\in\mathbf{F}$ by
$$D_xf(\wt{N})=f(\wt{N}+\delta_x)-f(\wt{N}),$$
and inductively define
$$D^n_{x_1,\ldots,x_n}f(\wt{N})=D_{x_n}^1D_{x_1,\ldots,x_{n-1}}^{n-1}f(\wt{N}),$$
with $D^0f=f$. The operator $D^n_{x_1,\ldots,x_n}$ is symmetric in $(x_1,\ldots,x_n)$ and the map $(x_1,\ldots,x_n,\wt{N})\mapsto D^n_{x_1,\ldots,x_n}f(\wt{N})$ is measurable. Define
the symmetric and measurable function
$$T_nf(x_1,\ldots,x_n)=\Expc D^n_{x_1,\ldots,x_n}f(\wt{N}),$$
with $T_0f=\Expc f(\wt{N})$. If $f,g\in L^2(\PN)$ with $\PN=\mathbb{P}(\wt{N}\in\cdot)$ then the following isometry holds
$$\Expc f(\wt{N})g(\wt{N})=\Expc f(\wt{N})\Expc g(\wt{N})+\sum_{n=1}^\infty\frac{1}{n!}\expc{T_nf,T_ng}_n$$
showing that $L^2(\PN)$ is isometric to a direct sum over all $L^2(\nu^n)$ (with rescaled inner products).
For simplicity let us write $L^2(\wt{N})$ for $L^2(\PN)$. Then for $g\in L^2(\wt{N})$ (or more generally $g\in L^1(\wt{N})$), it is possible with some care to define multiple stochastic integrals
$$I_n(g)=\int g(x_1,\ldots,x_n)d\wt{N}(x_1,\ldots,x_n)$$
Then, analogously to the Wiener case, we have the following orthogonality relation 
\begin{equation}
    \Expc I_n(f)I_m(g)=\delta_{nm}n!\expc{\wt{f},\wt{g}}_n
    \label{ortho}
\end{equation}
 for $f\in L^2(\nu^n)$, $g\in L^2(\nu^m)$ and with $\widetilde{h}$ denoting the symmetrization of $h\in L^2(\nu^p)$
$$\wt{h}=\frac{1}{p!}\sum_{\sigma\in\mathfrak{S}_p}h(x_{\sigma(1)},\ldots,x_{\sigma(p)})$$
over all permutations in the permutation group $\mathfrak{S}_p$ (see \cite{RW97} for an account of the orthogonality relation and multiple stochastic integrals). This leads to the 
chaos expansion, which is a unique representation of $f(\wt{N})$ for $f\in L^2(\wt{N})$ (so that $T_nf\in L^2(\nu^n)$),
$$f(\wt{N})=\sum_{n=0}^\infty I_n(T_n f),$$
where $I_0(T_0f)=\Expc f(\wt{N})$. We call the $n$-th chaos $\mathcal{C}_n=\mbox{span}\{I_n(T_nf)~:~f\in L^2(\wt{N})\}$ so that
$$L^2(\wt{N})\cong\bigoplus_n\mathcal{C}_n.$$
\par
The Malliavin derivative, which for Poisson measures coincides with the difference operator $D_x$, acts as a lowering operator on the chaoses, that is it acts on a random variable
$$F=\Expc F+\sum_{n=1}^\infty I_n(f_n)$$
such that 
$$\sum_{n=0}^\infty nn!\norm{f_n}_n^2<\infty$$
by
$$D_xF=\sum_{n=0}^\infty nI_{n-1}(f_n(x,\cdots)),$$
(see Theorem 3.3 of \cite{LastPen11}).
It is not too difficult to see that there is a special class of random variables which are invariant under this action
$$\EE(f)=1+\sum_{n=1}^\infty\frac{1}{n!} I_n(f^{\otimes_n}),$$
so that
$$D_x\EE(f)=f(x)\EE(f),$$
which shows that if we write $\mathfrak{E}=\{\EE(f)~:~f\in L^2(\nu)\}$, then $\mbox{span}\{\mathfrak{E}\}$ is invariant under $D_x$. These are directly related to the exponential vectors of 
Fock space in quantum mechanics. We have that $\mathfrak{E}$ is a total subset of $L^2(\wt{N})$, (a proof can be found in \cite{Nual06} for the Wiener case, which carries over with little 
modification to this case). We may construct processes 
by taking $(T,\mathcal{T},\nu)$ to be a time domain with a corresponding measure and by restricting the functions $f$ to $f_t=f\ident_{[0,t]}$, so that the integration is not over all $T$ but only
up to a certain time $t$. Then differential of $I_n\left(f^{\otimes_n}_t\right)$ is
$$dI_n\left(f^{\otimes_n}_t\right)=nf_tI_{n-1}\left(f^{\otimes_{n-1}}_t\right)d\wt{N}(t)$$
so that for $\EE(f_t)$ we have
\begin{equation}
    d\EE(f_t)=f_t\EE(f_t)d\wt{N}(t). 
    \label{doldadeeqn}
\end{equation}
This is the SDE satisfied by the Dol\'eans-Dade exponential process \cite{Prot04}. Note that $\EE(f)$ is a random variable, and $\EE(f_t)$ is a stochastic process defined via the filtration through
the conditional expectation,
$$\EE(f_t)=\Expc[\EE(f)|\filt_t].$$
Here $\EE(f)=\EE\left(\dint{T}{}fd\wt{N}\right)$ but that we sometimes omit the integral in order to coincide with the notation for exponential vectors.
If $X_t$ and $Y_t$ are two general processes, then they satisfy a general product rule \cite{Prot04},
\begin{equation}
    \EE(X_t)\EE(Y_t)=\EE\left(X_t+Y_t+[X,Y]_t\right),
\end{equation}
where $[X,Y]_t$ is the quadratic covariation process.
For two Poissonian martingales $X_t=\dint{0}{t}f_sd\wt{N}(s)$ and $Y_t=\dint{0}{t}g_sd\wt{N}(s)$, we have
$$[X,Y]_t=\int_0^tf_sg_sN(ds)=\int_0^tf_sg_s(d\wt{N}(s)+\nu(ds)),$$
which may be compared with the Wiener case where if $X_t=\dint{0}{t}f_sdW(s)$ and $Y_t=\dint{0}{t}g_sdW(s)$ then $[X,Y]_t=\dint{0}{t}f_sg_s\mu(ds)$.
This means that the product rule for a compensated Poisson-driven Dol\'eans-Dade process is
\begin{equation}
    \EE(X_t)\EE(Y_t)=\exp\left(\int_0^tf_sg_s\nu(ds)\right)\EE\left(\int_0^tf_s+g_s+f_sg_sd\wt{N}(s)\right)
    \label{prodrule}
\end{equation}
\section{Quantum stochastic calculus}
\label{qscsection}
Attempts to incorporate noise into quantum mechanics go back at least to the 1960's, but the fully formed analogue of the classical 
It\^o calculus was first constructed by Hudson and Parthasarathy in 1984 \cite{HudPar84}. We will only explore the basic results of this beautiful theory and refer the 
reader to \cite{Par92} for a full account. The setting is a bosonic Fock space $\hilb=\Gamma(\mathfrak{h})$ (fermionic versions are also possible) over a one particle space $\mathfrak{h}$. 
If we think of $\mathfrak{h}=L^2(\rr_+;\cx^d)$, as $d$-dimensional vector-valued functions of time, then the Fock space factors into tensor products for disjoint intervals of time
\begin{equation}
    \hilb=\hilb_{[0,t_1)}\otimes\hilb_{[t_1,t_2)}\otimes\cdots,
    \label{tensorfact}
\end{equation}
where $\hilb_{[t_i,t_{i+1})}=L^2([t_i,t_{i+1});\cx^d)$. In Fock space we can define a set of vectors
$$E=\left\{\ket{e(f)}=(1,f,\frac{1}{\sqrt{2}}f\otimes f,\cdots,\dfrac{1}{\sqrt{n!}}f^{\otimes_n},\ldots)~:~f\in\mathfrak{h}\right\}$$ called exponential vectors.
They have a simple rule for inner products
$$\expc{e(f)|e(g)}=\exp\expc{f|g}$$
and a factorization property 
$$\ket{e(f\oplus g)}=\ket{e(f)}\otimes\ket{e(g)},$$
where $f,g$ come from orthogonal subspaces, e.g. $f=h\ident_{[0,t)}$ and $g=h\ident_{[t,\infty)}$ for some $h\in\mathfrak{h}$. Operations acting in $\mathfrak{h}$ can be `second quantized'
into operators acting on $\fock$. One such operator is the Weyl operator $W(u,U)$ second quantizing translation by $u\in\mathfrak{h}$, and the transformation given by unitary $U$ acting on 
$\mathfrak{h}$. This operator acts on exponential vectors, $\ket{e(v)}$ like
$$W(u,U)\ket{e(v)}=\exp\left(-\frac{1}{2}\norm{u}^2-\expc{u,Uv}\right)\ket{e(Uv+u)}.$$
We may construct two semigroups, $t\mapsto W(tu,I)$ and $t\mapsto W(0,e^{itH})$ for self-adjoint $H$, whose generators are respectively $p(u)$ and $\lambda(H)$. From 
$p(u)$ we construct an operator $q(u)=-p(iu)$, and use it to define the annihilation and creation operators 
$$a(u)=\frac{1}{2}(q(u)+ip(u)),~~~a^\dagger(u)=\frac{1}{2}(q(u)-ip(u)).$$
We may extend the definition of $\lambda(H)$, the conservation operator, to include bounded not necessarily self-adjoint operators $B$ by defining
$$\lambda(B)=\lambda\left(\frac{1}{2}\left(B+B^\dagger\right)\right)+i\lambda\left(\frac{1}{2i}\left(B-B^\dagger\right)\right).$$
Exponential vectors are total in Fock space and so their span serves as a good core for the domains of these operators. In particular, we have
\begin{align*}
    a(u)\ket{e(g)}&=\expc{u|g}\ket{e(g)}\\
    \bra{e(f)}a^\dagger(u)&=\expc{f|u}\bra{e(f)}\\
    \expc{e(f)|\lambda(H)|e(g)}&=\expc{f|Hg}\expc{e(f)|e(g)}
\end{align*}
If we let $\{\ket{\alpha}\}$ be an orthonormal basis of 
$\cx^d$, then $\ident_{[0,t]}\ket{\alpha}\in L^2(\rr_+;\cx^d)$ and we may construct the fundamental processes,
\begin{align}
    A_\alpha(t)&=a(\ident_{[0,t]}\bra{\alpha})\otimes I_{(t,\infty)},\nonumber\\
    A^\dagger_\alpha(t)&=a^\dagger(\ident_{[0,t]}\ket{\alpha})\otimes I_{(t,\infty)},\\
    \Lambda_\alpha^\beta(t)&=\lambda(\ident_{[0,t]}\ket{\beta}\bra{\alpha})\otimes I_{(t,\infty)}\nonumber,
\end{align}
where $I_{(t,\infty)}$ is the identity on $\hilb_{(t,\infty)}$.
Hudson and Parthasarathy defined integrals with respect to these processes starting from simple processes with respect to a partition of $\rr_+$. The Fock factorizability condition,
equation (\ref{tensorfact}), ensures that these increments are independent. A limit may then be taken where the partition spacing goes to zero and an integral may be defined, where
the differentials are the limits of the fundamental processes' increments. They then derived the following quantum It\^o rules for these differentials,
\begin{equation}
\begin{array}{ccc} 
          dA_\alpha(t)dA_\beta^\dagger(t)=\delta_{\alpha\beta}dt & dA_\alpha(t)d\Lambda^\beta_\gamma(t)=\delta_{\alpha\beta}dA_\gamma(t)\\ 
           d\Lambda^\alpha_\beta(t)dA_\gamma^\dagger(t)=\delta_{\beta\gamma}dA_\alpha^\dagger(t) & d\Lambda^\alpha_\beta(t)d\Lambda^\gamma_\mu(t)=\delta_{\beta\gamma}d\Lambda^\alpha_\mu(t),\\ 
\end{array}
\end{equation}
where the differentials satisfy
\begin{align}
    \expc{e(f)|d\Lambda_\alpha^\beta(t) e(g)}&=f^*_\beta(t)g_\alpha(t)dt\expc{e(f)|e(g)},\nonumber\\
    \expc{e(f)|dA_\alpha(t)e(g)}&=g_\alpha(t)dt\expc{e(f)|e(g)},
        \label{fundmatelem}\\
    \expc{e(f)|dA^\dagger_\alpha(t) e(g)}&=f^*_\alpha(t)dt\expc{e(f)|e(g)}.\nonumber
\end{align}
If we let the vacuum be denoted $\ket{\Omega}=\ket{e(0)}$ by virtue of these definitions 
\begin{align}
dA_\alpha(t)\ket{\Omega}=0\\
d\Lambda_\alpha^\beta(t)\ket{\Omega}=0.
\label{actonvac}
\end{align}

Quantum stochastic processes, considered as integrals of certain operator processes with respect to the fundamental processes, may now be manipulated as quantum stochastic differential equations
in a way analogous to the classical way. We may ask the following question: Under what conditions might a quantum 
stochastic process, i.e. an integral with respect to the fundamental processes, be unitary. A class of such equations is given by the Hudson-Parthasarathy equation,
\begin{equation}
    dU_t=\left(\sum_{\alpha}L_\alpha dA_\alpha^\dagger(t)-\sum_\beta L_\beta^\dagger S_{\beta\alpha}dA_\alpha(t)+\sum_\beta(S_{\beta\alpha}-I)d\Lambda_\alpha^\beta(t)-\left(iH+
    \dfrac{1}{2}\sum_\alpha L_\alpha^\dagger L_\alpha\right)dt\right)U_t.
    \label{hpeqn}
\end{equation}
whose solution $U_t$ is unitary when $L_\alpha$ are bounded, $H=H^\dagger$, and $\sum_\alpha S_{\alpha\beta}^*S_{\alpha\gamma}=\sum_\alpha S_{\beta\alpha}S_{\gamma\alpha}^*=\delta_{\beta\gamma}$.
Note that we have omitted the tensor product symbol between system operators and fields and consider the system operators to be dilated to the whole Hilbert space $\wt{\hilb}=\hilb_S
\otimes\Gamma(\mathfrak{h})$ when needed by tensoring with the identity of $\Gamma(\mathfrak{h})$ or vice versa.\par
We can construct classical distributions from spectral measures of quantum observables. An observable may be decomposed via a projection valued measure $\xi:\rr\mapsto\mathcal{P}(\mathfrak{h})$ 
taking sets $A\subset \rr$ to $\xi(A)$ a corresponding projection, where the set of projections acting on a Hilbert space $\hilb$ is denoted $\mathcal{P}(\mathfrak{h})$. This decomposition is 
the spectral one,
$$O=\int_\rr \lambda\xi(d\lambda).$$
We may construct a classical distribution $\mu$ from this observable and a chosen state $\ket{u}$ as follows
$$\mu(A)=\int_A\lambda\expc{u|\xi(d\lambda)|u}.$$
In his book \cite{Par92}, Parthasarathy found that if $X(A)=\lambda(\eta(A))$ for some projection valued measure $\eta(A)\in\mathcal{P}(\mathfrak{h})$, then with respect to the coherent 
state $\ket{\wt{e}(f)}=\exp\left(-\frac{1}{2}\norm{f}^2\right)\ket{e(f)}$, $X(A)$ will give rise to a classical Poisson distribution with mean $\expc{f|\eta(A)|f}$. If we restrict
ourselves to the projection valued measure $\eta_t=\ident_{[0,t]}$ then we may consider $X$ to be a quantum Poisson process $X([0,t])=X(t)$, which highlights a curious relationship 
between time in the classical sense and time as a projection valued measures as we needn't have chosen $\eta_t=\ident_{[0,t]}$ per se. He also showed that if the desired state to act on is 
the vacuum $\ket{\Omega}$, then we can use the Weyl operators to get an equivalent process $$\expc{\Omega|W(-f,I)\lambda(\eta_t)W(f,I)|\Omega}=\expc{\wt{e}(f)|\lambda(\eta_t)|\wt{e}(f)}.$$
The operator $W(-f,I)\lambda(\eta_t)W(f,I)=\lambda(\eta_t)+a(\eta_t f)+a^\dagger(\eta_t f)+\expc{f|\eta_t|f}$ and so it can be shown that
\begin{equation}
    dX_\alpha(t)=d\Lambda_\alpha^\alpha(t)+dA_\alpha(t)+dA^\dagger_\alpha(t)+dt,
\end{equation}
defined on the core $\mbox{span}(E)$, form a set of independent Poisson processes when acting on the vacuum. We will further explore 
this in the next section.
\section{An isomorphism of Hilbert spaces}
\label{isomorphsect}
Fock space is isomorphic to $L^2(M)$ for certain random measures $M$. Both these spaces are Hilbert spaces, with the latter having the inner
product 
$$(f,g)_{L^2(M)}=\Expc f^*(M)g(M),$$
In the case where $M=W$, the white noise measure, the 
isomorphism was explored by Segal \cite{Segal56}, building off of Wiener's work with the chaos decomposition \cite{Wiener38}. However, it was realized that a similar relationship holds 
for $M=\wt{N}$ \cite{Versh75,Surg84}, and in fact for general L\'evy processes, though they do not in general have chaos decompositions. Vershik and Tsilevich have compiled the details 
of this isomorphism in the general L\'evy case \cite{Versh03}. We call the isomorphism $\Theta:\Gamma(\mathfrak{h})\rightarrow L^2(M_{\mathfrak{h}})$, where we note that there 
are as many independent processes as there are Lindblad operators so we denote this by $M_{\mathfrak{h}}$ and call each individual process $M_\alpha$. For processes exhibiting the chaotic 
decomposition property this isomorphism will map the $n$-particle space into the $n$-chaos, for each $n$.
In particular, it will preserve exponential vectors,
$$\Theta\ket{e(f)}=\EE(f).$$ 
Since both these sets are total in their respective spaces, they are convenient for calculations involving $\Theta$. In particular, if we would like to show that the quantum 
Poisson process $X_\alpha$ is equivalent to the classical Poisson process, then we may consider how they act on these objects.
\begin{theorem}
    Under the isomorphism $\Theta$, $X_\alpha(t)$ is equivalent to the operator of multiplication by the Poisson process $N(t)=\Theta X(t)\Theta^{-1}$.
\end{theorem}
\begin{proof}
    Analogously to how Parthasarathy and Usha Devi handled the Wiener process, we will show this by computing the `matrix elements' with respect to exponential processes/vectors. 
    Using equation (\ref{fundmatelem}), and dropping for the moment the indices, in the Fock case we have
    $$\expc{e(f)|dX(t) e(g)}=\expc{e(f)|e(g)}(f_t^*g_tdt+f^*_tdt+g_tdt+dt).$$
    In $L^2(\wt{N})$, the corresponding calculation is
    \begin{align*}
        \Expc N_t\EE(f)^*\EE(g)&=\exp\expc{f|g}\Expc N_t\EE(f^*+g+f^*g)\\
        &=\exp\expc{f|g}\Expc (I_1(\ident_{[0,t]})+t)\left(1+\sum_{n=1}^\infty\frac{1}{n!}I_n\left((f^*+g+f^*g)^{\otimes_n}\right)\right),
    \end{align*}
    where we have used that $\wt{N}_t=\dint{}{}\ident_{[0,t]}(s)d\wt{N}_s=I_1(\ident_{[0,t]})=N_t-t.$
    We may now reduce this using the orthogonality property of stochastic integrals, equation (\ref{ortho}), and the fact that $\Expc I_n(f_n)=0$ for all $n$ to get
    \begin{equation}
        \Expc N_t\EE(f)^*\EE(g)=\exp\expc{f|g}\left(\int_0^tf_s^*g_sds+\int_0^tf_s^*ds+\int_0^tg_sds +t\right),
        \label{expcN}
    \end{equation}
    which shows that $N_t\cong\dint{0}{t}dX$, under $\Theta$. Taking the dependence on $\alpha$ into account presents no problem as integrals with respect to $N_\alpha(t)$
    will be independent of those with respect to $N_\beta(t)$ for different $\alpha,\beta$.
\end{proof}
\begin{figure}
\centering
\begin{tikzcd}
    L^2(\wt{N}_{\mathfrak{h}};\hilb_0)
    \arrow[r,"\wt{U}_t"]
    \arrow[d,"\Theta^{-1}"]
    &
    L^2(\wt{N}_{\mathfrak{h}};\hilb_0)
    \\
    \hilb_0\otimes\Gamma(\mathfrak{h})
    \arrow[r,"U_t"]
    &
    \hilb_0\otimes\Gamma(\mathfrak{h})
    \arrow[u,"\Theta"]
\end{tikzcd}
\caption{A diagram showing how unitary evolution is mapped to a stochastic evolution}
\label{cd1}
\end{figure}
The unitary evolution given by the HP equation will be mapped to a stochastic evolution in Poisson space, as shown in figure \ref{cd1}, so that the stochastic
evolution is expressed as $\wt{U}_t=\Theta U_t\Theta^{-1}$. While the unitary evolution in the quantum space is deterministic, we think of the evolution in Poisson space as
probabilistic. This amounts to taking a \emph{probabilistic interpretation} of the quantum evolution (see e.g. \cite{OQS206,Meyer06} and references therein). Interestingly, we could have just 
as easily mapped the quantum space to a Wiener space. As the quantum evolution is not dependent on these isomorphisms, this
shows that the dynamics is not truly dependent on whichever probabilistic interpretation we choose. In this situation, it does not make sense to attribute anything fundamental
to the `jump' or the `diffusion', they are merely two ways to view the underlying quantum evolution. This dichotomy mirrors the `particle' and `wave' interpretations of
the quantum state, where the discreteness of the particle, e.g. in direct detection, is modeled with a Poisson process, and the continuity of the wave, e.g. in homodyne detection, is modeled
with a Wiener process.

\section{Piecewise deterministic processes from HP evolution}
Starting from the HP equation (\ref{hpeqn}), we may reexpress the noises in terms of the quantum Poisson processes, $X_\alpha$, by adding and subtracting the corresponding terms.
This will result in a unitary adapted process which we can apply to an initial state $\ket{\phi_0\otimes\Omega}$ where $\ket{\phi_0}$ is the initial system state and $\ket{\Omega}$ is the
vacuum of Fock space. Using the isomorphism $\Theta$, or more accurately $I_S\otimes\Theta$ which we will denote by the same symbol, this will give the stochastic evolution of the 
state vector $\ket{\psi_t}=\Theta U_t\ket{\phi_0\otimes\Omega}$. Then from
the HP equation we can calculate the SDE that $\ket{\psi_t}$ must satisfy, $d\ket{\psi_t}=\Theta dU_t\ket{\phi_0\otimes\Omega}$. When acting on the vacuum as an initial condition, 
by virtue of equation (\ref{actonvac}), the HP equation can be greatly simplified:
\begin{align*}
    d\ket{\psi_t}&=\Theta\left(\sum_{\alpha}L_\alpha dX_\alpha-\left(\sum_\beta L_\beta^\dagger S_{\beta\alpha}+L_\alpha\right)dA_\alpha(t)+\sum_\beta(S_{\beta\alpha}-I-\delta_{\alpha\beta}L_\alpha)d\Lambda_\alpha^\beta(t)-
    L_\alpha dt\right.\\
    &\left.-\left(iH+\dfrac{1}{2}\sum_\alpha L_\alpha^\dagger L_\alpha\right)dt\right)U_t\ket{\phi_0\otimes\Omega}\\
    &=\Theta\left(\sum_\alpha L_\alpha dX_\alpha-L_\alpha dt-\left(iH+\frac{1}{2}\sum_\alpha L_\alpha^\dagger L_\alpha\right)dt\right)\Theta^{-1}\Theta U_t\ket{\phi_0\otimes\Omega}.\\
\end{align*}
This results in a linear SDE for $\ket{\psi_t}$, using the relationship $dN_{\alpha}=\Theta dX_{\alpha}\Theta^{-1}$,  driven by a compensated Poisson process 
$d\wt{N}_\alpha=dN_\alpha-dt$,
\begin{equation}
    d\ket{\psi_t}=\left(\sum_\alpha L_\alpha d\wt{N}_\alpha-\left(iH+\frac{1}{2}\sum_\alpha L_\alpha^\dagger L_\alpha\right)dt\right)\ket{\psi_t}.
    \label{linsde}
\end{equation}
It is a simple exercise to show that this is an unravelling of the GKSL equation, though it is not normalized. This equation does allow us to consider the norm-squared through
the SDE
$$d\normpsi=(d\bra{\psi_t})\ket{\psi_t}+\bra{\psi_t}(d\ket{\psi_t})+d\bra{\psi_t}d\ket{\psi_t}$$
This SDE can be solved.
\begin{lemma}
The norm-squared process satisfies the SDE,
$$d\normpsi=\sum_\alpha R_\alpha d\wt{N}_\alpha\normpsi,$$
where $R_\alpha=\dfrac{\expc{\psi_t\Big|L_\alpha+L_\alpha^\dagger+L_\alpha^\dagger L_\alpha\Big|\psi_t}}{\normpsi}$. Thus the norm-squared process is a Dol\'eans-Dade process driven
by the martingale $\wt{N}$,
$$\normpsi=\EE\left(\sum_\alpha\int_0^tR_\alpha d\wt{N}_\alpha\right).$$
\end{lemma}
\begin{proof}
This follows directly from It\^o calculus, recalling that the It\^o rules for Poisson processes are
\begin{align*}
    dN_\alpha dN_\beta&=\delta_{\alpha\beta}dN_\alpha,\\
    dN_\alpha dt&=0
\end{align*}
so that $d\wt{N}_\alpha d\wt{N}_\beta=\delta_{\alpha\beta}dN_\alpha=\delta_{\alpha\beta}(d\wt{N}_\alpha+dt)$. That the process is a Dol\'eans-Dade process can be seen by
comparison with equation (\ref{doldadeeqn} )
\end{proof}
To get a normalized unravelling we will have to invert this process and construct a process, $\Phi_t$, such that $\Phi_t^*\Phi_t=\inormpsi$. Then we can multiply the
unnormed solution by $\Phi_t$ to get a normalized solution $\ket{\Psi_t}=\Phi_t\ket{\psi_t}$. In order to do this it is necessary to invert the Dol\'eans-Dade process.
\begin{lemma}
    Let $X_t=\dint{0}{t}f d\wt{N}$.  Assume that with probability $1$ the jumps of $X_t$ are strictly greater than $-1$. Then the inverse of $\EE(X_t)$ is
$$\EE(X_t)^{-1}=\exp\left(\int_0^t\dfrac{f^2}{1+f}ds\right)\EE\left(\int_0^t\dfrac{-f}{1+f}d\wt{N}\right).$$
\label{inverse}
\end{lemma}
\begin{proof}
    From the product rule for Dol\'eans-Dade processes, equation (\ref{prodrule}), we have that if $\EE(Y_t)$ with $Y_t=\dint{0}{t}gd\wt{N}$ is to invert $\EE(X_t)$ then we must have
    $f+g+fg=0$ so that $g=\dfrac{-f}{1+f}$. The exponential term then comes from cancelling out the term that comes from reexpressing the quadratic variation in 
    terms of the compensated Poisson process. We will have $\EE(X_t)\EE(X_t)^{-1}=\EE(0)=1$.
\end{proof}
Letting $S_\alpha=\frac{R_\alpha}{1+R_\alpha}$, then using the preceding lemma we may express the inverse norm-squared process as
$$\inormpsi=\exp\left(\sum_\alpha\int_0^tR_\alpha S_\alpha ds\right)\EE\left(-\sum_\alpha\int_0^tS_\alpha d\wt{N}_\alpha\right).$$
Note that $R_\alpha=\dfrac{\expc{\psi_t\big|(L_\alpha^\dagger+I)(L_\alpha+I)\big|\psi_t}}{\expc{\psi_t|\psi_t}}-1\ge -1$, with equality only in the case $\expc{\psi_t\big|(L_\alpha^\dagger+I)(L_\alpha+I)\big|\psi_t}=0$.
The latter possibility will be addressed after equation (\ref{girsanovthm}).
The inverse norm-squared satisfies the following SDE,
$$d\inormpsi=\left[\sum_\alpha R_\alpha S_\alpha dt-S_\alpha d\wt{N}_\alpha\right]\inormpsi.$$
Multiplication of the inverted norm-squared process by the unnormalized density matrix $\rho_\psi(t)=\state{\psi_t}$ results in a normalized density matrix $\rho_\Psi(t)=\state{\Psi_t}$. 
We can think of the norm-squared process as the Radon-Nikodym derivative for a change of measure. Then we have
$$\Expc\rho_\psi(t)=\Expc\normpsi\rho_\Psi(t)=\Expc^\prime\rho_\Psi(t),$$
where we have changed from a probability space $\mathcal{P}=(\Omega,\filt,\filt_t,\mathbb{P})$ to the probability space $\mathcal{P}^\prime=(\Omega,\filt,\filt_t,\mathbb{Q})$, denoting
the expected value with respect to $\mathbb{Q}$ as $\Expc^\prime$. The Radon-Nikodym derivative is then 
$\normpsi=\Expc\left[\dfrac{d\mathbb{Q}}{d\mathbb{P}}\bigg|\filt_t\right]$. The Girsanov-Meyer theorem for jump processes gives new martingales in the primed probability space, 
\begin{equation}
\wt{N}_\alpha^\prime=\wt{N}_\alpha-\dint{0}{t}\expc{\psi_{s-}|\psi_{s-}}^{-1} d\langle\wt{N}_\alpha,\expc{\psi_s|\psi_s}\rangle=\wt{N}_\alpha-\dint{0}{t}R_\alpha ds,
\label{girsanovthm}
\end{equation}
(see e.g. \cite{Prot04}) where $\langle\wt{N}_\alpha,\expc{\psi_s|\psi_s}\rangle$ is the angle-bracket process which is the compensator of the quadratic covariation process $[\wt{N}_\alpha,\expc{
\psi_s|\psi_s}]$. This will allow us to calculate the new expected values, $\Expc^\prime d\wt{N}_\alpha=\Expc^\prime R_\alpha dt$ and $\Expc^\prime dN_\alpha=\Expc^\prime(1+R_\alpha)dt$.
Note that in the case $R_\alpha=-1$, the probability of a jump goes to zero, so the jump part of the integral in Lemma \ref{inverse} is actually a.s. zero, leaving only the time part of the
compensated Poisson integral, which when combined with the coefficient process will take care of any potential singularities.
\par
The normalized density matrix, $\rho_\Psi(t)$, will satisfy the SDE
\begin{align*}
d\rho_\Psi(t) &= \lind[\rho_\Psi(t)]dt+\left(\sum_\alpha R_\alpha S_\alpha dt-S_\alpha d\wt{N}_\alpha\right)\rho_\Psi(t)\\
&+\left(\sum_\alpha L_\alpha\rho_\Psi(t)+\rho_\Psi(t)L_\alpha^\dagger+L_\alpha\rho_\Psi(t)L^\dagger_\alpha\right)(d\wt{N}_\alpha-S_\alpha dN_\alpha),
\end{align*}
which can be shown by applying the It\^o product rule to $\rho_\Psi(t)=\inormpsi\state{\psi_t}$. In $\mathcal{P}^\prime$, this density matrix will integrate the master equation 
in the sense that $\Expc^\prime d\rho_\Psi(t)=\lind\left[\Expc^\prime \rho_\Psi(t)\right]dt$. However, in order to get the benefit of an unravelling, an SDE 
for the state vector is required to reduce the number of dimensions. Having the inverse norm-squared process, one can find the process, $\Phi_t$, which will normalize the state vector, 
$\ket{\psi_t}$.
\begin{lemma}
    Let $c_\alpha=\dfrac{1}{\sqrt{1+R_\alpha}}$ and define
    \begin{equation}
        \Phi_t=\exp\left(\sum_\alpha\int_0^t\left[\frac{1}{2}(R_\alpha-2)+c_\alpha\right] ds\right)\EE\left(\sum_\alpha\int_0^t\left(c_\alpha-1\right) d\wt{N}_\alpha\right).
    \end{equation}
    Then $\Phi_t^*\Phi_t=\inormpsi$.
    \label{philemma}
\end{lemma}
\begin{proof}
    This follows again from the product rule for Dol\'eans-Dade processes. We are looking for $\Phi_t$ to be in the form $\Phi_t=\mathcal{N}_t\EE(X_t)$, with 
    $X_t=\sum_\alpha\dint{0}{t}f_\alpha d\wt{N}_\alpha$ an unknown martingale, $f_\alpha$ a real valued adapted process, and $\mathcal{N}_t$ a factor which we can use to make the coefficient 
    processes work out. 
    The requirement is then
    $$|\mathcal{N}_t|^2\EE(X_t)^*\EE(X_t)=\mathcal{M}_t\EE(Y_t)$$
    where 
    $$\mathcal{M}_t=\exp\left(\sum_\alpha\int_0^tR_\alpha S_\alpha ds\right).$$
    Then $f_\alpha$ must satisfy $2f_\alpha+f_\alpha^2=-S_\alpha$. Simple algebra requires $f_\alpha=-1\pm\sqrt{1-S_\alpha}=-1\pm c_\alpha$, with $c_\alpha=\dfrac{1}{\sqrt{1+R_\alpha}}$.
    We will take $f_\alpha=-1+c_\alpha$, so to compensate $\mathcal{M}_t$ and
    the terms that come from reexpressing the quadratic variation term, $1-2c_\alpha+c_\alpha^2$ in terms of the compensated Poisson process, we must have 
    $$\mathcal{N}_t=\exp\left(\sum_\alpha\int_0^t\frac{1}{2}(R_\alpha-2)+c_\alpha ds\right),$$
    This shows that $\Phi_t$ has the required form.
\end{proof}
Having the form of $\Phi_t$, we can write down the SDE it satsifies,
\begin{equation}
    d\Phi_t=\left[\sum_\alpha\left(\frac{1}{2}(R_\alpha-2)+c_\alpha\right)dt+(c_\alpha-1)d\wt{N}_\alpha\right]\Phi_t.
    \label{phisde}
\end{equation}
Finally, the SDE for the normalized unravelling can be obtained.
\begin{theorem}
    The process $\ket{\Psi_t}=\Phi_t\ket{\psi_t}$, satisfying the SDE
    \begin{align}
        d\ket{\Psi_t}=-\left(iH+\frac{1}{2}\right.&\left.\sum_\alpha L_\alpha^\dagger L_\alpha+2L_\alpha+1-(R_\alpha+1)\right)dt\ket{\Psi_t}\nonumber\\
        &+\sum_\alpha\left(c_\alpha(L_\alpha+I)-I\right)dN_\alpha\ket{\Psi_t},
        \label{normedsde}
    \end{align}
    is a normalized unravelling of the GKSL equation, in the sense that $$\Expc^\prime d\left(\state{\Psi_t}\right)=\lind\left[\Expc^\prime \state{\Psi_t}\right]dt.$$
\end{theorem}
\begin{proof}
    That $\ket{\Psi_t}$ is an unravelling follows from Lemma \ref{philemma} and the Girsanov theorem defining the change of measure, making this unravelling equivalent to the
    unnormalized one. That $\ket{\Psi_t}$ satisfies
    equation \ref{normedsde} follows from the It\^o product rule, using equations \ref{linsde} and \ref{phisde}.
\end{proof}
This completes the derivation of a normalized unravelling from a unitary evolution where the bosonic reservoir has been interpreted through
the isomorphism, $\Theta$, as resulting in Poissonian noise. This equation has a different form then equation \ref{PDP}, but we have the following corollary.
\begin{corollary}
    The normalized unravelling SDE can be written as the canonical PDP,
    \begin{align}
        d\ket{\Psi_t}&=-\left(iH^\prime+\frac{1}{2}\sum_\alpha\left( M_\alpha^\dagger M_\alpha-\norm{M_\alpha\Psi_t}^2\right)\right)dt\ket{\Psi_t}\\
        &+\sum_\alpha\left(\frac{M_\alpha}{\norm{M_\alpha\Psi_t}}-I\right)dN_\alpha \ket{\Psi_t}.
        \label{canonpdpeqn}
    \end{align}
    \label{canonpdp}
\end{corollary}
\begin{proof}
The substitution $M_\alpha=L_\alpha+I$ gives the desired equation. Note that $\lind$ and the HP equation are invariant under the transformations $L_\alpha \mapsto L_\alpha+f_\alpha I$ and 
$H\mapsto H+\frac{1}{2i}\sum_\alpha f_\alpha^*L_\alpha-f_\alpha L_\alpha^\dagger$, for constant (or generally locally time square-integrable) functions $f_\alpha$ \cite{ParDevi17}. By making this 
substitution and defining $H^\prime=H+\frac{1}{2i}\sum_\alpha M_\alpha-M_\alpha^\dagger$, we obtain an equivalent unravelling of the GKSL equation.
\end{proof}
We further note that if $\sum_\alpha M_\alpha^\dagger M_\alpha=1$ as is done in GRW or if $M_\alpha$ are projectors into the eigenspaces of a self-adjoint operator, 
the resulting equation simplifies to
\begin{align}
        d\ket{\Psi_t}&=-iH\ket{\Psi_t}dt+\sum_\alpha\left(\frac{M_\alpha}{\norm{M_\alpha\Psi_t}}-I\right)dN_\alpha \ket{\Psi_t},
        \label{qmsleq}
\end{align}
which is a Schr\"odinger evolution augmented by, in the case of eigenprojetors, jumps into the eigenspaces defined by $M_\alpha$ and where the probability, $p_\alpha(t)$, of a jump into the 
$\alpha$-th eigenspace in the time interval $[t,t+dt]$ is given by the rate of the Poisson process, $p_\alpha(t)dt$. Notably, it will satisfy the Born rule 
$p_\alpha(t)=\norm{M_\alpha\Psi_t}^2$.

\section{A note on the Girsanov transformation}

In order that the normalized purely system state evolution unravel the GKSL equation, we changed the original probability measure $\mathbb{P}$ to a new probability measure $\mathbb{Q}$ 
(Girsanov transformation), using the norm-squared process as the Radon-Nikodym derivative.  The operator $G[\psi_t]$ of multiplication by $\Phi_t$ establishes a unitary isomorphism of the Hilbert 
spaces $L^2(\wt{N})$ and $L^2(\wt{N}^\prime)$. 
To see this we denote the inner product in the range space by $(\cdot,\cdot)^\prime$, and the inner product in the domain as $(\cdot,\cdot)$.  We have:
\begin{align*}
    (Gf,Gg)^\prime&=\left(\Phi_tf,\Phi_tg\right)^\prime=\Expc^\prime\inormpsi\expc{f|g}\\
    &=\Expc\normpsi\inormpsi\expc{f|g}=\Expc\expc{f|g}=(f,g)
\end{align*}
Just as we did for the inverse norm-squared process we can explicitly calculate the square root of the norm-squared process to get the norm process, $\phi_t$.
\begin{proposition}
The square root of the process $\normpsi$ is the process $\phi_t$ given by
\begin{equation}
\phi_t=\exp\left(-\frac{1}{2}\sum_\alpha\int_0^t\left(\norm{M_\alpha\Psi_s}-1\right)^2ds\right)\EE\left(\sum_\alpha\int_0^t\norm{M_\alpha\Psi_s}-1d\wt{N}_\alpha\right).
\label{normeq}
\end{equation}

\end{proposition}
\begin{proof}
The calculation to obtain this expression is exactly analogous to Lemma \ref{philemma}. Alternatively one can use the inversion formula of Lemma \ref{inverse} on the process
$\Phi_t$.
\end{proof}
\begin{figure}
\centering
\begin{tikzcd}
    L^2(\wt{N}_{\mathfrak{h}};\hilb_0)
    \arrow[r,"\wt{U}_t"]
    & 
    L^2(\wt{N}_{\mathfrak{h}};\hilb_0)
    \arrow[r,"G{[\psi_t]}"]
    &
    L^2(\wt{N}_{\mathfrak{h}}^\prime;\hilb_0)
    \arrow[d,"{\Theta^\prime}^{-1}"]
    \\
    \hilb_0\otimes\Gamma(\mathfrak{h})
    \arrow[r, "U_t"] 
    \arrow[u,"\Theta"]
    & 
    \hilb_0\otimes\Gamma(\mathfrak{h})
    \arrow[r, "\widehat{G}{[\psi_t]}"]
    \arrow[u,"\Theta"]
    & 
    \hilb_0\otimes\Gamma(\mathfrak{h}) 
\end{tikzcd}
\caption{Diagram of dynamical relationship and normalization procedure}
\label{cd2}
\end{figure}

Just as the space $L^2(\wt{N})$ is the space of square-integrable functionals of the compensated Poisson process $\wt{N}$, elements of $L^2(\wt{N}^\prime)$ are square-integrable functionals 
of the compensated Poisson process $\wt{N}'$, obtained from $\wt{N}$ by the Girsanov theorem.  Both $L^2$ spaces are isomorphic to the quantum state space  $\hilb_0\otimes\Gamma(\mathfrak{h})$ via 
chaos decomposition of functionals of the Poisson processes.  The multiplication operator $G[\psi_t]$ introduced above can thus be represented as a unitary automorphism of this space.  More precisely,  
denoting by $\Theta':  \hilb_0\otimes\Gamma(\mathfrak{h}) \to  L^2(\wt{N}'_{\mathfrak{h}};\hilb_0)$ the unitary isomorphism between the quantum state space and the space of square-integrable 
functionals of the new (compensated) Poisson process, and recalling the isomorphism $\Theta:  \hilb_0\otimes\Gamma(\mathfrak{h}) \to  L^2(\wt{N}_{\mathfrak{h}};\hilb_0)$, we define the map
$$
\widehat{G}[\psi_t]={\Theta^{\prime}}^{-1}G[\psi_t]\Theta.
$$
$\widehat{G}[\psi_t]$ is a unitary automorphism of the quantum state space, as diagrammed in figure \ref{cd2}.  We emphasize, that $\widehat{G}[\psi_t]$ depends on the choice of the initial 
condition in the HP equation \ref{hpeqn}.
Since the 
inverse norm process $\Phi_t$ is a stochastic integral with an integrand which is random, it does not in general have a simple chaos decomposition.  Consequently, the corresponding operator 
$\widehat{G}[\psi_t]$ is difficult to describe explicitly.  It is the probabilistic representation of the bosonic Fock space that makes the definition of $\widehat{G}[\psi_t]$ possible and natural. 
In the 
context of Wiener chaos representation, this idea was advanced by Parthasarathy and Usha Devi \cite{ParDevi17}, defining a ``generalized Weyl process'' and using it in conjunction with the 
HP evolution to solve the Gisin-Percival equation.  The fruitful idea of generalizing operators in the Fock space using its probabilistic representation is worth exploring further.  

\section{Applications of the stochastic Schr\"odinger equation}

Equation (\ref{qmsleq}) is a stochastic Schr\"odinger equation.  For a certain choice of operators $M_\alpha$ and, with the discrete index $\alpha$ replaced by a continuous parameter, 
it may be thought of as a SDE describing the GRW interpretation \cite{GRW85}.  The random process driving the dynamics leads to the dynamical derivation of the standard Born rule.  Thus 
equation (\ref{qmsleq})
may be thought of either from the perspective of open quantum systems, as an unravelling of a particular Lindblad equation, or, from the perspective of GRW as a fundamental equation which
ultimately gives rise to the consistency of our observable world through a dynamical selection of random measurement results.  
\par
For either interpretation, it is important to derive equations for the evolution of expected values of observables. For any observable $O$ we have
to the equation
$$d\expc{\Psi_t|O|\Psi_t}=(d\bra{\Psi_t})O\ket{\Psi_t}+\bra{\Psi_t}O(d\ket{\Psi_t})+(d\bra{\Psi_t})O(d\ket{\Psi_t})$$
Performing some cancellations, one obtains
$$d\expc{\Psi_t|O|\Psi_t}=\expc{\Psi_t|i[H,O]|\Psi_t}+\sum_\alpha\left( \dfrac{\expc{\Psi_t|M_\alpha OM_\alpha|\Psi_t}}{\norm{M_\alpha\ket{\Psi_t}}^2}-\expc{\Psi_t|O|\Psi_t}\right)dN_\alpha.$$
When $O =H$, the Hamiltonian, the commutator is zero and, taking into account the distribution of $N_{\alpha}$, it follows after some simplification that
$$\Expc d\expc{\Psi_t|H|\Psi_t}=\sum_\alpha\left(\Expc\expc{\Psi_t|M_\alpha HM_\alpha|\Psi_t}-\Expc\expc{\Psi_t|H|\Psi_t}\right)dt.$$
The first term is diagonal in the eigenbasis of the $M_\alpha$'s in the case that $M_\alpha$ are eigenprojectors of a self-adjoint operator, while the second term is then responsible for decay 
of the off-diagonal terms, leading to decoherence. It is a linear first-order inhomogeneous differential equation and so if we denote $\expc{H}_t=\Expc\expc{\Psi_t|H|\Psi_t}$ and
$\expc{H^\prime}_t=\Expc\expc{\Psi_t|M_\alpha HM_\alpha|\Psi_t}$, it has the solution
$$\expc{H}_t=\int_0^te^{s-t}\expc{H^\prime}_sds+\expc{H}_0e^{-t},$$
showing that after large times the system loses memory of its original expected energy and the projected term $\expc{H^\prime}_t$ becomes dominant.
\par
The master equation of GRW (see \cite{BG03} for an in-depth treatment) is 
$$\dfrac{d\rho(t)}{dt}=-i[H,\rho(t)]+\lambda\left(T[\rho(t)]-\rho(t)\right)$$
where $\lambda$ is the rate of jumps and 
$$T[\rho]=\int d^3\mathbf{x} L_{\mathbf{x}}\rho L_{\mathbf{x}}$$
where their localization operator, $L_{\mathbf{x}}$, is defined as
$$L_{\mathbf{x}}=\left(\dfrac{a}{\pi}\right)^{3/4}e^{-(a/2)(\hat{q}-\mathbf{x})^2}$$
for some value of the phenomenological parameter $a$ (typically $\alpha$ is used), where $\hat{q}$ denotes the position operator. The parameter $\mathbf{x}$ in the GRW formulation plays the role 
analogous to the index $\alpha$ in the present article, labelling possible localizations of the particle in space.  Note that $\dint{}{}d^3\mathbf{x}L_{\mathbf{x}}^2=1$.  Since we are using a 
standard Poisson process, the analogue of $\lambda$ is equal to $1$ in our case.  In the GRW formulation, $\norm{L_{\mathbf{x}}\ket{\Psi_t}}^2$ represents the probability density that a jump localizes 
the system around position $\mathbf{x}$.  Again, up to changing a discrete index for a continuous parameter, this is captured by equation (\ref{qmsleq}) of the present work.
\par
Going further, it is possible to describe the requirements for how a GRW unravelling may come about. We may consider a sequence of points $\{\mathbf{x}_\alpha\}$ which in the limit of 
$\alpha\rightarrow\infty$ densely fill the space $\rr^3$. We can define intermediary Lindblad operators 
$$L_{\alpha}^{(k)}=\left(\dfrac{a}{\pi}\right)^{3/4}e^{-(a/2)(\hat{q}-\mathbf{x}_\alpha)^2}\sqrt{\delta_k}$$
where $\delta_k$ is a small parameter playing the role of the differential $d\mathbf{x}$ in the limit so that $\sum_{\alpha=1}^k {L_\alpha^{(k)}}^2\rightarrow\dint{}{}d^3
\mathbf{x}L_{\mathbf{x}}^2=1$ as $k\rightarrow\infty$. For each $k$ we have
a HP solution $U_t^{(k)}$ for the set of $k$ Lindblad operators $\{L_\alpha^{(k)}\}$, which we can use to derive an intermediary GRW unravelling. We describe this in the
following theorem, and note that there is one caveat pertaining to the procedure of going from $L_\alpha^{(k)}+I$ to $L_\alpha^{(k)}$ but we will see this can be addressed.
\begin{theorem}
Consider all process to be defined on the compact interval $[0,T]$. Let $\{L_\alpha^{(k)}\}$ be the Lindblad localization operators associated the sequence $\{\mathbf{x}_\alpha\}$ which fills 
$\mathbb{R}^3$. Then if $\sum_\alpha {L_\alpha^{(k)}}^2\rightarrow I$,
and we choose noises $A_\alpha^{(k)}(t)=a(\bra{c^{(k)}_\alpha}\ident_{[0,t]})$ and ${A^{(k)}}^\dagger_\alpha(t)=a^\dagger(\ket{c_\alpha^{(k)}}\ident_{[0,t]})$, 
with $\{\ket{c_{\alpha}^{(k)}}\}$ real and orthonormal for each $k$, such that
\begin{equation}
\sum_{\alpha}L_\alpha^{(k)}\sum_\beta L_\beta^{(l)}\expc{c_\alpha^{(k)}|c^{(l)}_\beta}\rightarrow I
\label{lcovar}
\end{equation}
as $k,l\rightarrow\infty$,
then the sequence of solutions to the HP equations $U_t^{(k)}$ corresponding to the Lindblad operators $\{L_\alpha^{(k)}\}$ (with no scattering terms) strongly converges.
\label{grwlimit}
\end{theorem}
\begin{proof}
Let $\ket{\xi}=\ket{\phi\otimes e(u)}\in\hilb_0\otimes\Gamma(\mathfrak{h})$, where we will need $\mathfrak{h}$ to be $L^2(\rr_+)\otimes \hilb_1$ where $\hilb_1$ is a Hilbert space of 
countable dimension taken to be real for convenience, and $\phi\in\hilb_0$, $u\in\mathfrak{h}$. We've chosen $\ket{\xi}$ to be a member of a dense set in $\hilb_0\otimes\Gamma(\mathfrak{h})$. Then
the condition for strong convergence is that $\forall\epsilon>0$ $\exists N\in\mathbb{N}$ such that if $k,l>N$ we have
\begin{align*}
    \norm{(U^{(k)}_t-U^{(l)}_t)\ket{\xi}}^2=2\expc{\xi|\xi}-\expc{\xi\bigg|{U^{(k)}}^\dagger_t U^{(l)}_t\bigg|\xi}-\expc{\xi|{U_t^{(l)}}^\dagger U_t^{(k)}\bigg|\xi}<\epsilon
\end{align*}
where we have used that $U_t^{(k)}$ and $U^{(l)}_t$ are solutions to HP equations and therefore unitary.
For strong convergence then we require that $\expc{\xi\bigg|{U^{(k)}}^\dagger_t U^{(l)}_t\bigg|\xi}$ and its conjugate must converge to $\expc{\xi|\xi}$. If we examine the differential 
$d\expc{\xi\bigg|{U^{(k)}}^\dagger_t U^{(l)}_t\bigg|\xi}$, we can use the It\^o calculus to show that
\begin{align}
d\expc{\xi\bigg|{U^{(k)}}^\dagger_t U^{(l)}_t\bigg|\xi}&=\expc{\xi\bigg|\left[{U^{(k)}_t}^\dagger\sum_\alpha L_\alpha^{(k)}(d{A_\alpha^{(k)}}-d{A_\alpha^{(k)}}^\dagger)U^{(l)}_t\right]\bigg|\xi}
\nonumber\\
&+\expc{\xi\bigg|U^{(k)}\left[\sum_\beta L_\beta^{(l)}(d{A_\beta^{(l)}}^\dagger-dA_\beta^{(l)})\right]U^{(l)}_t\bigg|\xi}\nonumber\\
&+\expc{\xi\bigg|{U_t^{(k)}}^\dagger\left[\sum_\alpha L_\alpha^{(k)}\sum_\beta L_\beta^{(l)}\expc{c_\alpha^{(k)}|c_\beta^{(l)}}-\frac{1}{2}\sum_\alpha {L_\alpha^{(k)}}^2-\frac{1}{2}
\sum_\beta {L_\beta^{(l)}}^2\right]U^{(l)}_t\bigg|\xi}dt
\label{ukul}
\end{align}
where we have explicitly used the fact that the Lindblad operators are self-adjoint.
The first term in equation \ref{ukul} can be evaluated as
$$\sum_\alpha\expc{\xi\bigg|{U^{(k)}_t}^\dagger L_\alpha^{(k)}U_t^{(l)}\bigg|\xi}(\expc{u|c_\alpha}-\expc{c_\alpha|u})dt$$
and we see that as $\hilb_1$ was taken to be real, this term is zero, and similarly for the second term in equation \ref{ukul}. For the last term, we can choose an $N$ for each $\epsilon^\prime=
\dfrac{\epsilon}{T}$ 
such that for all $k,l>N$ we have, via the Cauchy-Shwarz inequality, that 
$$\bigg|\expc{\zeta\bigg|I-\sum_\alpha L_\alpha^{(k)}\sum_\beta L_\beta^{(l)}\expc{c_\alpha^{(k)}|c_\beta^{(l)}}\bigg|\eta}\bigg|\le\frac{\epsilon^\prime}{6}\expc{\zeta|\zeta}^{1/2}
\expc{\eta|\eta}^{1/2}$$
and 
$$\bigg|\expc{\zeta\bigg|I-\sum_\alpha {L_\alpha^{(k)}}^2\bigg|\eta}\bigg|\le\frac{\epsilon^\prime}{3}\expc{\zeta|\zeta}^{1/2}\expc{\eta|\eta}^{1/2},~~~\mbox{and}~~~~~\bigg|\expc{\zeta\bigg|I-\sum_\alpha 
{L_\alpha^{(l)}}^2\bigg|\eta}\bigg|\le\frac{\epsilon^\prime}{3}\expc{\zeta|\zeta}^{1/2}\expc{\eta|\eta}^{1/2}$$
for $\ket{\zeta},\ket{\eta}\in\hilb_0\otimes\Gamma(\mathfrak{h})$, because these sums are positive contractions. Taking $\bra{\zeta}=\bra{\xi U^{(k)}_t}$ and $\ket{\eta}=\ket{U^{(l)}_t\xi}$, 
and writing $X_t=\expc{\xi\bigg|{U^{(k)}_t}^\dagger U^{(l)}_t\bigg|\xi}$ we have upon integration that
$$|X_t-X_0|\le \dfrac{\epsilon^\prime}{2}t\expc{\xi|\xi}.$$
It follows that $X_t\ge-\dfrac{\epsilon^\prime}{2}t\expc{\xi|\xi}+\expc{\xi|\xi}$, using the fact that $U_0^{(k)}=U_0^{(l)}=I$. Doing the same for $X_t=\expc{\xi\bigg|{U_t^{(l)}}^\dagger U_t^{(k)}
\bigg|\xi}$,
we have
$$\norm{\left(U_t^{(k)}-U_t^{(l)}\right)\ket{\xi}}^2=2\expc{\xi|\xi}-2\expc{\xi|\xi}+t\epsilon^\prime\expc{\xi|\xi}<\epsilon\expc{\xi|\xi}.$$
\end{proof}
As we are free to choose whatever noises we want at each step $k$, so long as $\{\ket{c^{(k)}_\alpha}\}$ are orthonormal for each $k$, there is nothing in principle preventing us from 
finding a set of $\{\mathbf{x}_\alpha\}$, $\{L_\alpha^{(k)}\}$, $\{\ket{c_\alpha^{(k)}}\}$ and $\{\delta_k\}$ which satisfy the requirements of Theorem \ref{grwlimit}, however we leave the details for a future work. 

In the derivation of the stochastic Schr\"odinger equation which gives rise to the GRW unravelling, we made the substitution $L_{\alpha} + I = M_{\alpha}$.  This may seem problematic, since when the 
number of $\alpha$ is infinite, the sum of squares of the Lindblad operators no longer converges to a bounded operator.  It is however possible to make sense of this procedure in the following way.  
Let us denote the unitary process solving the HP equation with Lindblad operators  $L_{\alpha}$ by $U_t$.  If finitely many operators $L_{\alpha}$ are replaced by $L_{\alpha} + f_{\alpha}(t)$,
with $f_\alpha(t)$ a locally time square integrable function, the unitary process $V_t$  solving the corresponding modified HP equation is related to the previous one by $V_t = R_t U_t$, where 
$R_t=I_S\otimes W(\ident_{[0,t]}\sum_\alpha f_\alpha(t)\ket{\alpha})$ acts as the identity on the system space, and $W$ is the Weyl operator (see \cite{ParDevi17} for details).  
It follows that the Heisenberg evolution of a system observable $X$,
$$
j_t(X) = V^{\dagger}_tX \otimes I_EV_t = U^{\dagger}_tR^{\dagger}_tX \otimes I_ER_tU_t = U^{\dagger}_tX \otimes I_EU_t,
$$
 is unchanged under this transformation. Since this is true when any finite number of $L_{\alpha}$ is changed, it is natural to extend this by definition to infinitely many such changes. The GRW 
 unravelling is then selected from the equivalence class of unravellings which give rise to this Heisenberg evolution.
This procedure can also be done for non-Hermitian Lindblad operators and, as this has been proposed as a solution the unbounded energy increase in GRW \cite{Smir14}, this is certainly worth 
considering, although our proof does not directly apply to this case. 

\par
As a second application we recall that the unravellings of the Lindblad equation can be used as a computational method for integrating the Lindblad equation,
where, instead of having $n^2$ coupled ODEs for
the density matrix, we have $n$ coupled SDEs to solve. In the case of Poissonian noise, this algorithm dates back to the Monte Carlo wave function (MCWF) approach of Dalibard, Castin, and M\o lmer
\cite{DCM92}. Their algorithm may be described as follows \cite{MCD93}:
\begin{itemize}
\item[1.] Given the state at time $t$, $\ket{\psi_t}$, calculate the evolution of the state under the non-Hermitian Hamiltonian 
$$H=H_S-i\dfrac{\hbar}{2}\sum M_\alpha^\dagger M_\alpha,$$
which for a small enough time step $\delta t$ is approximately given by
$$\ket{\psi^0_{t+\delta t}}=\left(I-\dfrac{iH\delta t}{\hbar}\right)\ket{\psi_t}.$$
As this evolution is not Hermitian, norm is not preserved. The norm equals
$$\expc{\psi^0_{t+\delta t}|\psi^0_{t+\delta}}=\expc{\psi_t\bigg|\left(1+\dfrac{iH^\dagger\delta t}{\hbar}\right)\left(1-\dfrac{iH\delta t}{\hbar}\right)\bigg|\psi_t}=1-\delta p$$
with 
$$\delta p=\delta t\expc{\psi_t|H^\dagger-H|\psi_t}=\sum_\alpha \delta p_\alpha$$
and $\delta p_\alpha=\delta t\expc{\psi_t|M_\alpha^\dagger M_\alpha|\psi_t}\ge 0$. Adjust $\delta t$ so this equation is valid to first order in $\delta t$; we require $\delta p\ll 1$.
\item[2.] Decide whether a jump happens by choosing a random number $\epsilon$ from the uniform distribution on $[0,1]$. 
\begin{itemize}
    \item[i.] If $\delta p<\epsilon$, there is no jump. Renormalize the state vector found in step 1:
    $$\ket{\psi_{t+\delta t}}=\dfrac{\ket{\psi^0_{t+\delta t}}}{(1-\delta p)^{1/2}}$$
    \item[ii.] If $\delta p>\epsilon$ a jump occurs. Choose a jump according to the probability $\Pi_\alpha=\dfrac{\delta p_\alpha}{\delta p}$, so choose $\epsilon^\prime\in [0,1]$ and if
    $\sum_{\beta=0}^{\alpha-1}\Pi_{\beta}<\epsilon^\prime<\sum_{\beta=0}^{\alpha}\Pi_\beta$, with $\Pi_0=0$ and $\alpha=1,\ldots n$ if there are $n$ Lindblad operators, then jump corresponding 
    to the operator $M_{\alpha}$ occurs 
    and the new state is $$\ket{\psi_{t+\delta t}}=\dfrac{M_\alpha\ket{\psi_t}}{(\delta p_\alpha/\delta t)^{1/2}}$$
\end{itemize}
\end{itemize}
Observe that in the limit of small $\delta t$ (and so small $\delta p$) this algorithm reduces to the unravelling in Corollary 1. If we relate $d\ket{\psi_t}$ to $\ket{\psi_{t+\delta t}}-\ket{\psi_t}$
then if $\delta p$ is small we have $\dfrac{1}{(1-\delta p)^{1/2}}\approx 1-\dfrac{1}{2}\delta p$ and steps 1 and 2.i give the $dt$ term of equation \ref{canonpdpeqn}. The probability $\delta p_
\alpha$, that a jump corresponding to the operator $M_{\alpha}$ will occur in a time interval of length $\delta t$ becomes the jump rate $\norm{M_\alpha\psi_t}^2dt$ of the Poisson process of 
equation \ref{canonpdpeqn} for small $\delta t$.  The two terms in the coefficient of the Poisson process in equation \ref{canonpdpeqn} describe the state change of step 2.ii
Thus we can see that the unravelling which we have derived here from QSC considerations can be seen as the small $\delta t$ limit of the MCWF algorithm.  It is natural to conjecture that the 
(random) sequence of states generated by the algorithm converges to the solution of the stochastic Schr\"odinger equation.  We do not address this point which goes beyond the application of the 
algorithm to solving Lindblad equation.

\section{GRW from QSC and the measurement problem}

The spontaneous localization theories by Ghirardi, Rimini, Weber, Pearle \cite{GRW85,GPR90} and others rely on {\it postulating} nonlinearity and randomness in the evolution of a quantum 
system.  In the present work we use a linear unitary evolution defined by the HP equation to {\it derive} a stochastic nonlinear dynamics, which we then claim can be used to motivate
GRW, and related models.  In doing this, we follow Parthasarathy and Usha Devi \cite{ParDevi17}, who realized such a program in the form of a stochastic differential equation driven by Wiener 
processes.  In our case, they are replaced by  Poisson processes.
\par
In both cases, randomness appears thanks to an isomorphism $\Theta$ of the bosonic Fock space, representing the environment, with an appropriate space of random variables.  This maps the joint 
unitary evolution of the system and the environment to a stochastic evolution of the system's state, which does not preserve its norm.  It is thus the map $\Theta$ that generates the randomness 
which appears in the Born rule.  It is important to point out that the quantum noises used here can be obtained as a limit of real quantum fields, as detailed in the book by Accardi et al 
\cite{Accardi13}, and so this randomness need not be taken as fundamental, in GRW or otherwise. 
\par
In the case of GRW, nonlinearity is introduced by normalization, but in our work there is a Girsanov map $G$ which is an isomorphism between 
$L^2(\wt{N})$ and $L^2(\wt{N}')$ and so the normalization procedure is not nonlinear in the conventional sense. In fact, we may recover the norm given the normalized state (equation \ref{normeq}) 
so the normalization is entirely reversible. However, the Girsanov map depends on the initial state of the system and so the normalization procedure is conditioned on this initial state.
This in particular means that a superposition of two vector-valued processes in $L^2(\wt{N};\hilb_0)$ is not mapped to a superposition of normalized-vector-valued processes in $L^2(\wt{N}';\hilb_0)$ in 
accordance with Bassi and Ghirardi's argument against linearity \cite{BG00}.
\par 
We emphasize that a single vector in the Fock space is mapped by $\Theta$ onto a random variable---a functional of the Poisson process (or  Wiener process in \cite{ParDevi17}).  The evolution of the 
system is thus represented as a solution of the resulting SDE.  It is a random process, and for a fixed time---a random variable with values in the system's state space.  The procedure does 
not provide a correspondence between individual realizations of the solution of the SDE and vectors in the system's Hilbert space.  A parallel may be drawn with the Everettian interpretation here 
in that the quantum picture is that of a state which contains all information about possibilities, however we do not have a superposition of different possibilities, as Everett envisioned, since in 
this picture there is no sense in which different trajectories exist together; rather they are all contained in the single quantum state through the introduction of a probabilistic picture.

The two probabilistic representations of the environment---based on Wiener and Poisson processes---while mathematically equally valid (in fact, isomorphic), lead to very different physical 
interpretations of quantum evolution, in particular the measurement process.  While the Wiener picture describes continuous acquisition of information, in the Poisson case information is acquired 
in discrete portions, (see \cite{Bar06} for details on measurement schemes and the HP equation).  This mirrors the wave-particle complementarity and merits further exploration.  
Here we make no detailed statement as to the physical situations corresponding to the two 
probabilistic pictures.  We stress however, that either one of them can introduce more clarity into the description of the quantum evolution, as evidenced by the simplicity of the operator 
$G[\psi_t]$ versus the intractability of $\widehat{G}_t[\psi_t]$.  Mathematically, this is originating from additional structure of the spaces of random variables, in particular the presence of 
pointwise multiplication, absent in the Fock space.  

Lastly, there is another feature shared between this picture and Everett's interpretation---appearance of a part of the wave function which keeps track of the 
history of the system. Everett \cite{Ev57} introduced an observer part of the wave function which keeps track of the results of measurements along the branches, and which also has the effect of 
conditioning the branches on the trajectory of the system, as does $\widehat{G}[\psi_t]$. The norm process, $\phi_t$, has a similar property in that it encodes the history of the system into its state, 
and also has special mathematical structure---that of a (multiple of) an exponential martingale.  This structure is not accidental or artificial; rather, it appears
naturally as a result of the probabilistic interpretation of the quantum evolution.  The information contained in the norm process makes it a kind of universal observer which is entirely quantum,
appearing as part of the total quantum state, as Everett originally envisioned.

\section{Concluding remarks}

Poisson-driven unravellings of master equations were derived from HP evolution, and, with a particular choice of Lindblad operators,  a comparison was drawn between these unravellings and the
dynamics of the GRW model of quantum mechanics, a proposed solution to the measurement problem. The primary advantage of our approach is that the noise---considered classical in GRW---is derived 
here from the unitary quantum evolution, with no added dynamical features (but within the quantum noise model of the interaction with the environment).  Since the pioneering GRW work, 
stochastic collapse models have moved on to CSL and even relativistic equations \cite{BP99}. CSL uses a version of the Gisin-Percival equation which has already been explored from the perspective 
of HP evolution by Parthasarathy, Usha Devi, \cite{ParDevi17}, Barchielli and Belavkin \cite{Bel90,BarBel91}, the latter also deriving Poisson unravellings to slightly different ends with
Barchielli \cite{BarBel91} and Staszewski \cite{BelStas91}. It may also be possible to derive the relativistic models from a unitary evolution in a similar manner. In this case, rather than 
time dependent unitary groups, they would depend on proper time and the HP equation would not suffice since proper-time-dependent Lindblad operators are required.
This is a promising direction for future research, and a motivation to consider unitary quantum stochastic evolutions beyond the HP equation.
\begin{acknowledgements}
The authors were partially supported by the National Science Foundation grant DMS 1615045.  D.K. gratefully acknowledges Michael Tabor Fellowship from the Program of Applied Mathematics, 
University  of Arizona.  Both authors thank the Institute of Photonic Science, Castelldefels (Spain) and Maciej Lewenstein for their hospitality. 
\end{acknowledgements}
\bibliographystyle{abbrv}
\bibliography{article}
\end{document}